\documentclass[a4paper,UKenglish,cleveref, autoref, thm-restate,authorcolumns]{lipics-v2019}

\def\MdN{\ensuremath{\mathbb{N}}}
\newcommand{\Oh}[1]{\mathcal{O}\!\left( #1\right)}
\newcommand{\Is}{:=}
\newcommand{\etal}{et~al.\ }
\newcommand{\ie}{i.e.\ }
\newcommand{\prOne}{\texttt{HeavyEdge}}
\newcommand{\prTwo}{\texttt{ImbalancedVertex}}
\newcommand{\prThree}{\texttt{ImbalancedTriangle}}
\newcommand{\prFour}{\texttt{HeavyNeighborhood}}
\newcommand{\CC}{C\texttt{++}}

\newcommand{\optZero}{\texttt{BasicCactus}}
\newcommand{\optOne}{\texttt{+Connectivity}}
\newcommand{\optTwo}{\texttt{+LocalContract}}
\newcommand{\optThree}{\texttt{+DegreeOne}}
\newcommand{\optFour}{\texttt{+C\&LInCactus}}
\newcommand{\optFive}{\texttt{+D1InCactus}}
\newcommand{\optSix}{\texttt{FullAlgorithm}}

\newif\ifEnableExtend
\EnableExtendtrue

\newif\ifDoubleBlind
\DoubleBlindfalse

\usepackage{algorithm}
\usepackage[noend]{algpseudocode}
\usepackage{graphicx}
\usepackage{numprint}
\usepackage[section]{placeins}
\usepackage{subcaption}

\let\oldReturn\Return
\renewcommand{\Return}{\State\oldReturn}

\title{Finding All Global Minimum Cuts In Practice}

\ifDoubleBlind
\author{Double Blind}{Double Blind University}{}{}{}
\else\author{Monika Henzinger}{University of Vienna, Faculty of Computer Science, Vienna, Austria}{monika.henzinger@univie.ac.at}{0000-0002-5008-6530}{}
\author{Alexander Noe}{University of Vienna, Faculty of Computer Science, Vienna, Austria}{alexander.noe@univie.ac.at}{0000-0002-4711-3323}{}
\author{Christian Schulz}{University of Vienna, Faculty of Computer Science, Vienna, Austria}{christian.schulz@univie.ac.at}{0000-0002-2823-3506}{}
\author{Darren Strash}{Hamilton College, Clinton, NY, USA}{dstrash@hamilton.edu}{0000-0001-7095-8749}{}
\fi{}

\nolinenumbers

\hideLIPIcs  

\EventEditors{John Q. Open and Joan R. Access}
\EventNoEds{2}
\EventLongTitle{42nd Conference on Very Important Topics (CVIT 2016)}
\EventShortTitle{CVIT 2016}
\EventAcronym{CVIT}
\EventYear{2016}
\EventDate{December 24--27, 2016}
\EventLocation{Little Whinging, United Kingdom}
\EventLogo{}
\SeriesVolume{42}
\ArticleNo{23}



\category{}

\relatedversion{}

\authorrunning{M. Henzinger, A. Noe, C. Schulz and D. Strash}
\Copyright{Monika Henzinger, Alexander Noe, Christian Schulz and Darren Strash}

\ifDoubleBlind

\else
\funding{
  The research leading to these results has received funding from the European Research Council under the
European Community's Seventh Framework Programme (FP7/2007-2013) /ERC grant agreement No. 340506. \newline Partially supported by DFG grant SCHU 2567/1-2.}
\fi{}

\begin{document}

\maketitle

\begin{abstract}
We present a practically efficient algorithm that finds all global minimum cuts in huge undirected graphs. Our algorithm uses a multitude of kernelization rules to reduce the graph to a small equivalent instance and then finds all minimum cuts using an optimized version of the algorithm of Nagamochi, Nakao and Ibaraki. In shared memory we are able to find all minimum cuts of graphs with up to billions of edges and millions of minimum cuts in a few minutes. 
\ifEnableExtend We also give a new linear time algorithm to find the most balanced minimum cuts given as input the representation of all minimum cuts.\fi{}
\end{abstract}

\ifDoubleBlind
\newpage
\setcounter{page}{1}
\fi

\section{Introduction}

We consider the problem of finding \emph{all minimum cuts} of an undirected network where edges are weighted by positive integers. A \emph{minimum cut} in a graph is a partition of the vertices into two sets so that the total weight of edges crossing the boundary between the blocks is minimized. The problem of finding all minimum cuts has applications in many fields. In particular, minimum cuts in similarity graphs can be used to find clusters~\cite{wu1993optimal,hartuv2000clustering}. As the minimum cut is often highly skewed, a variety of techniques to find more balanced bipartitions with small cuts were developed~\cite{ding2001min,hagen1992new,shi2000normalized}. However, these balanced variations of the problem are generally NP-complete. In contrast to that, we find all minimum cuts in practice in a similar timescale than finding an arbitrary minimum cut and can thus output the minimum cut that is least skewed. In community detection, the absence of a small cut inside a cluster can indicate a likely community in a social network~\cite{cai2005mining}.
Other applications for finding all minimum cuts can be found in network reliability~\cite{karger2001randomized,ramanathan1987counting}, where a minimum cut in a network has the highest risk to disconnect the network; in VLSI design~\cite{krishnamurthy1984improved} and graph drawing~\cite{kant1993algorithms}, in which minimum cuts are used to separate the network; and finding all minimum cuts is an important subproblem for edge-connectivity augmentation algorithms~\cite{gabow1991applications,naor1997fast}.

The problem of finding all minimum cuts is closely related to the \emph{(global) minimum cut problem}, which aims to find \emph{some} minimum cut in the graph. The classical algorithm of Gomory and Hu~\cite{gomory1961multi} solves the minimum cut problem by solving $n$ minimum-$s$-$t$ cut problems. Using the push-relabel algorithm~\cite{goldberg1988new} this results in a total running time of $\Oh{n^2m\log{\frac{n^2}{m}}}$. Nagamochi~\etal\cite{nagamochi1992computing,nagamochi1994implementing} give an algorithm for the minimum cut problem, which is based on edge contractions instead of maximum flows. Their algorithm has a worst case running time of $\Oh{nm+n^2\log{n}}$ but performs far better in practice on many graph classes~\cite{Chekuri:1997:ESM:314161.314315,junger2000practical,henzinger2018practical}. Henzinger~\etal\cite{henzinger2019shared} give a fast shared-memory parallel algorithm for the minimum cut problem based on their algorithm which finds a minimum cut very fast in practice.

Even though a graph can have up to $\binom{n}{2}$ minimum cuts~\cite{karger2000minimum}, there is a compact representation of all minimum cuts of a graph called \emph{cactus graph} with $\Oh{n}$ vertices and edges. A cactus graph is a graph in which each edge belongs to at most one simple cycle. Karzanov and Timofeev~\cite{karzanov1986efficient} give the first polynomial time algorithm to construct the cactus representation for all minimum cuts. Picard and Queyranne~\cite{picard1980structure} show that all minimum cuts separating two specified vertices can be found from a maximum flow between them. Thus, similar to the classical algorithm of Gomory and Hu~\cite{gomory1961multi} for the minimum cut problem, we can find all minimum cuts in $n-1$ maximum flow computations. The algorithm of Karzanov and Timofeev~\cite{karzanov1986efficient} combines all those minimum cuts into a cactus graph representing all minimum cuts. Nagamochi and Kameda~\cite{nagamochi1994canonical} give a representation of all minimum cuts separating two vertices $s$ and $t$ in a so-called $(s,t)$-cactus representation. Based on this $(s,t)$-cactus representation, Nagamochi~\etal\cite{nagamochi2000fast} give an algorithm that finds all minimum cuts and gives the minimum cut cactus in $\Oh{nm + n^2 \log{n} + n^*m\log{n}}$, where $n^*$ is the number of vertices in the cactus.

Karger and Stein~\cite{karger1996new} give a randomized algorithm to find all minimum cuts in $\Oh{n^2 \log^3{n}}$ time by contracting random edges. Based on the algorithm of Karzanov and Timofeev~\cite{karzanov1986efficient} and its parallel variant given by Naor and Vazirani~\cite{naor1991representing} they show how to give the cactus representation of the graph in the same asymptotic time. Ghaffari~\etal\cite{ghaffari2019faster} recently gave an algorithm that finds all \emph{non-trivial minimum cuts} of a simple unweighted graph in $\Oh{m \log^2{n}}$ time. Using the techniques of Karger and Stein the algorithm can trivially give the cactus representation of all minimum cuts in $\Oh{n^2 \log{n}}$. 
While there are multiple implementations of the algorithm of Karger and Stein~\cite{Chekuri:1997:ESM:314161.314315,gianinazzi2018communication,henzinger2018practical} for the minimum cut problem, to the best of our knowledge there are no published implementations of either of the algorithms to find the cactus graph representing all minimum cuts\ifEnableExtend (with or without data reduction techniques)\fi{}.

\ifEnableExtend
In the last two decades significant advances in FPT algorithms have been made: an NP-hard graph problem is fixed-parameter tractable (FPT) if large inputs can be solved efficiently and provably optimally, as long as some problem parameter is small. This has resulted in an algorithmic toolbox that are by now well-established. 
Few of the new techniques are implemented and tested on real datasets, and their practical potential is far from understood. 
However, recently the engineering part in area has gained some momentum. 
There are several experimental studies in the area that take up ideas from FPT or kernelization theory, e.g.~for independent sets (or equivalently vertex cover)~\cite{akiba-tcs-2016,DBLP:conf/sigmod/ChangLZ17,dahlum2016accelerating,DBLP:conf/alenex/Lamm0SWZ19,DBLP:journals/corr/abs-1908-06795,DBLP:conf/alenex/Hespe0S18}, for cut tree construction\cite{DBLP:conf/icdm/AkibaISMY16}, for treewidth computations \cite{bannach_et_al:LIPIcs:2018:9469,DBLP:conf/esa/Tamaki17,koster2001treewidth}, for the feedback vertex set problem \cite{DBLP:conf/wea/KiljanP18,DBLP:conf/esa/FleischerWY09}, for the dominating set problem~\cite{10.1007/978-3-319-55911-7_5}, for the minimum cut~\cite{henzinger2019shared,henzinger2018practical}, for the multiterminal cut problem~\cite{henzinger2019sharedmemory}, for the maximum cut problem~\cite{DBLP:journals/corr/abs-1905-10902} and for the cluster editing problem~\cite{Boecker2011}.
Recently, this type of data reduction techniques is also applied for problem in P such as matching \cite{DBLP:conf/esa/KorenweinNNZ18},
\fi{}

\paragraph*{Our Results}

We give an algorithm that finds all minimum cuts on very large graphs. Our algorithm is based on the recursive algorithm of Nagamochi~\etal\cite{nagamochi2000fast}. We combine the algorithm with a multitude of techniques to find edges that can be contracted without affecting any minimum cut in the graph. Some of these techniques are adapted from techniques for the global minimum cut problem~\cite{padberg1991branch,nagamochi1994implementing}. Using these and newly developed reductions we are able to decrease the running time by up to multiple orders of magnitude compared to the algorithm of Nagamochi~\etal\cite{nagamochi2000fast} and are thus able to find all minimum cuts on graphs with up to billions of edges in a few minutes. Based on the cactus representation of all minimum cuts, we are able to find the most balanced minimum cut in time linear to the size of the cactus. As our techniques are able to find the most balanced minimum cut\ifEnableExtend (and all others too)\fi{} of graphs with billions of edges in mere minutes, this allows the use of minimum cuts as a subroutine in sophisticated data mining and graph analysis tools.

\section{Basic Concepts}

Let $G = (V, E, c)$ be a weighted undirected simple graph with vertex set $V$, edge set $E \subset V \times V$ and
non-negative edge weights $c: E \rightarrow \MdN$. 
We extend $c$ to a set of edges $E' \subseteq E$ by summing the weights of the edges; that is, let $c(E')\Is \sum_{e=(u,v)\in E'}c(u,v)$ and let $c(u)$ denote the sum of weights of all edges incident to vertex $v$.
Let $n = |V|$ be the
number of vertices and $m = |E|$ be the number of edges in $G$. The \emph{neighborhood}
$N(v)$ of a vertex $v$ is the set of vertices adjacent to $v$. The \emph{weighted degree} of a vertex is the sum of the weights of its incident edges. For brevity, we simply call this the \emph{degree} of the vertex.
For a set of vertices $A\subseteq V$, we denote by $E[A]\Is \{(u,v)\in E \mid u\in A, v\in V\setminus A\}$; that is, the set of edges in $E$ that start in $A$ and end in its complement.
A cut $(A, V
\setminus A)$ is a partitioning of the vertex set $V$ into two non-empty
\emph{partitions} $A$ and $V \setminus A$, each being called a \emph{side} of the cut. The \emph{capacity} or \emph{weight} of a cut $(A, V
\setminus A)$ is $c(A) = \sum_{(u,v) \in E[A]} c(u,v)$.
A \emph{minimum cut} is a cut $(A, V
\setminus A)$ that has smallest capacity $c(A)$ among all cuts in $G$. We use $\lambda(G)$ (or simply
$\lambda$, when its meaning is clear) to denote the value of the minimum cut
over all $A \subset V$. For two vertices $s$ and $t$, we denote $\lambda(G,s,t)$ as the capacity of the smallest cut of $G$, where $s$ and $t$ are on different sides of the cut. $\lambda(G,s,t)$ is also known as the \emph{minimum s-t-cut} of the graph. If all edges have weight $1$, $\lambda(G,s,t)$ is also called the \emph{connectivity} of vertices $s$ and $t$. The connectivity $\lambda(G,e)$ of an edge $e=(s,t)$ is defined as $\lambda(G,s,t)$, the connectivity of its incident vertices. At any point in the execution of a minimum cut algorithm,
$\hat\lambda(G)$ (or simply $\hat\lambda$) denotes the smallest upper bound of the
minimum cut that the algorithm discovered until that point. 
For a vertex $u \in V$ with minimum vertex degree, the size of the \emph{trivial cut} $(\{u\}, V\setminus \{u\})$ is equal to the vertex degree of $u$.
Many algorithms tackling the minimum cut problem use \emph{graph contraction}.
Given
an edge $e = (u, v) \in E$, we define $G/(u, v)$ (or $G/e$) to be the graph after \emph{contracting
edge} $(u, v)$. In the contracted graph, we delete vertex $v$ and all edges
incident to this vertex. For each edge $(v, w) \in E$, we add an edge $(u, w)$
with $c(u, w) = c(v, w)$ to~$G$ or, if the edge already exists, we give it the edge
weight $c(u,w) + c(v,w)$.

A graph with $n$ vertices can have up to $\Omega(n^2)$ minimum cuts~\cite{karger2000minimum}. To see that this bound is tight, consider an unweighted cycle with $n$ vertices. Each set of $2$ edges in this cycle is a minimum cut of $G$. This yields a total of $\binom{n}{2}$ minimum cuts. 
However, all minimum cuts can be represented by a cactus graph $C_G$ with up to $2n$ vertices and $\Oh{n}$ edges~\cite{nagamochi2000fast}. A cactus graph is a connected graph, in which any two simple cycles have at most one vertex in common. In a cactus graph, each edge belongs to at most one simple cycle. 

To represent all minimum cuts of a graph $G$ in an edge-weighted cactus graph $C_G = (V(C_G), E(C_G))$, each vertex of $C_G$ represents a possibly empty set of vertices of $G$ and each vertex in $G$ belongs to the set of one vertex in $C_G$. Let $\Pi$ be a function that assigns to each vertex of $C_G$ it set of vertices of $G$. Then every cut $(S, V(C_G) \backslash S)$ corresponds to a minimum cut $(A, V \backslash A)$ in $G$ where $A=\cup_{x\in S} \Pi(x)$. In $C_G$, all edges that do not belong to a cycle have weight $\lambda$ and all cycle edges have weight $\frac{\lambda}{2}$. A minimum cut in $C_G$ consists of either one tree edge or two edges of the same cycle. We denote by $n^*$ the number of vertices in $C_G$ and $m^*$ the number of edges in $C_G$. The weight $c(v)$ of a vertex $v \in C_G$ is equal to the number of vertices in $G$ that are assigned to $v$.

\section{Algorithm Description}

Our algorithm combines a variety of techniques \ifEnableExtend and algorithms\fi{} in order to find all minimum cuts in a graph. The algorithm is based on the contractions of edges which cannot be part of a minimum cut. Thus, we first show that an edge $e$ that is not part of any minimum cut in graph $G$ can be contracted and all minimum cuts of $G$ remain in the resulting graph $G/e$.

\begin{lemma} \label{lem:contractible}
If an edge $e=(u,v)$ is not part of any minimum cut in graph $G$, all minimum cuts of $G$ remain in the resulting graph $G/e$.
\end{lemma}

\begin{proof}
  Let $(A,B)$ be an arbitrary minimum cut of $G$. For an edge $e = (u,v)$, which is not part of any minimum cut, we know that $e \not\in E[A]$, so either $u$ and $v$ are both in vertex set $A$ or both in vertex set $B$. This is still the case in $G/e$. Thus, the edge $e$ can be contracted even if we aim to find every minimum cut of $G$.
\end{proof}

Lemma~\ref{lem:contractible} is very useful to reduce the size of the graph by a multitude of techniques to identify such edges. We first give a short overview of our algorithm and then explain the techniques in more detail. First, we use the shared-memory parallel heuristic minimum cut algorithm VieCut~\cite{henzinger2018practical} in order to find an upper bound $\hat\lambda$ for the minimum cut which is very likely to be the correct value.\ifEnableExtend VieCut is a multilevel algorithm that uses the label propagation algorithm to contract heavily connected clusters.\fi{} Having a tight bound for the minimum cut allows the contraction of many edges, as multiple reduction techniques depend on the value of the minimum cut. We adapt contraction techniques originally developed by Nagamochi~\etal\cite{nagamochi1992computing,nagamochi1994implementing} and Padberg~\etal\cite{padberg1991branch} to the problem of finding all minimum cuts. Section~\ref{ss:contraction} explains these contraction routines. On the resulting graph we find all minimum cuts using an optimized variant of the algorithm of Nagamochi, Nakao and Ibaraki~\cite{nagamochi2000fast} and return the cactus graph which represents them all. A short description of the algorithm and an explanation of our engineering effort are given in Section~\ref{ss:cactus}. Afterwards, in Section~\ref{ss:combine} we show how we combine the parts into a fast algorithm to find all minimum cuts of large networks. 

\subsection{Edge Contraction}
\label{ss:contraction}

As shown in Lemma~\ref{lem:contractible}, edges that are not part of any minimum cut can be safely contracted. We build a set of techniques that aim to find contractible edges and run these in alternating order until neither of them finds any more contractible edges. We now give a short introduction to these.

For efficiency, we perform contractions in bulk. If our algorithm finds an edge that can be contracted, we merge the incident vertices in a thread-safe union-find data structure~\cite{anderson1991wait}. After each run of a contraction technique that finds contractible edges, we create the contracted graph using a shared-memory parallel hash table~\cite{maier2016concurrent}. In this contracted graph, each set of vertices of the original graph is merged into a single node. The contraction of this vertex set is equivalent to contracting a spanning tree of the set. After contraction we check whether a vertex in the contracted graph has degree $< \hat\lambda$. If it does, we found a cut of smaller value and update $\hat\lambda$ to this value.

\subsubsection{Connectivity-based Contraction}
\label{sss:connectivity}

The connectivity of an edge $e=(s,t)$ is the weight of the minimum cut that separates $s$ and $t$, \ie the \emph{minimum s-t-cut}. For an edge that has connectivity $> \hat\lambda$, we thus know that there is no cut separating $s$ and $t$ (\ie no cut that contains $e$) that has value $\leq \hat\lambda$. Thus, we know that there cannot be a minimum cut that contains $e$, as $\hat\lambda$ is by definition at least as large as $\lambda$. However, solving the minimum s-t-cut problem takes significant time, so computing the connectivity of each edge does not scale to large networks. 
Hence, as part of their algorithm for the global minimum cut problem, Nagamochi~\etal\cite{nagamochi1992computing,nagamochi1994implementing} give an algorithm that computes a lower bound $q(e)$ for the connectivity of every edge $e$ of $G$ in a total running time of $\Oh{m + n\log{n}}$. Each of the edges whose connectivity lower bound is already larger than $\hat\lambda$ can be contracted as it cannot be part of any minimum cut. Their algorithm builds \emph{edge-disjoint maximum spanning forests} and contracts all edges that are not in the first $\lambda - 1$ spanning forests, as those connect vertices that have connectivity at least $\lambda$~\cite{henzinger2019shared}. This is possible as the incident vertices of any such edge $e$ are connected in each of the first $\lambda - 1$ spanning forests and by $e$ and thus have a connectivity of at least $\lambda$. In other words, there can not be any cut smaller than $\lambda$ which contains $e$.

 Henzinger~\etal\cite{henzinger2019shared} give a fast shared-memory parallel variant of their algorithm. As both of these algorithms only aim to find a single minimum cut, they also contract edges that have connectivity equal to $\hat\lambda$, as they only want to see whether there is a cut better than the best cut known previously. As we want to find all minimum cuts, we can only contract edges whose connectivity is strictly larger than $\hat\lambda$. Nagamochi~\etal could prove that at least one edge has value $\hat\lambda$ in their routine and can thus be contracted. We do not have such a guarantee when trying to find edges that have connectivity $> \hat\lambda$. Consider for example an unweighted tree, whose minimum cut has a value of $1$ and each edge has connectivity $1$ as well.

 \subsubsection{Local Contraction Criteria}
 \label{sss:local}
 
 \begin{figure*}[t!]
   \centering
   \includegraphics[width=.8\textwidth]{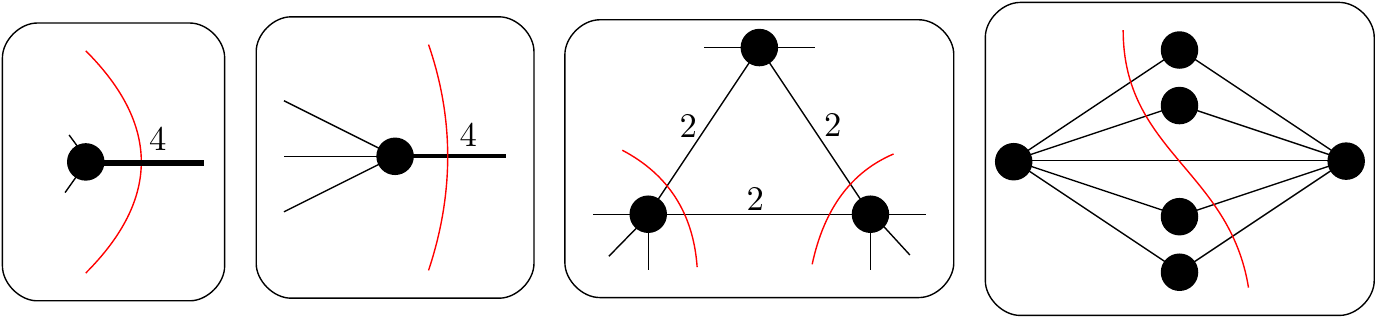}
   \caption{Contraction: (1) \prOne, (2) \prTwo, (3) \prThree, (4) \prFour}
   \label{fig:local}
 \end{figure*}

 Padberg and Rinaldi~\cite{padberg1991branch} give a set of \emph{local reduction} routines which determine whether an edge can be contracted without affecting the minimum cut. Their reductions routines were shown to be very useful in order to find a minimum cut fast in practice~\cite{Chekuri:1997:ESM:314161.314315,junger2000practical,henzinger2018practical}. We adapt the routines \ifEnableExtend originally developed for the minimum cut problem so that they hold for the problem of \fi{} for finding all minimum cuts. Thus, we have to make sure that we do not contract cuts of value $\hat\lambda$, as they might be minimal and additionally make sure that we do not contract edges incident to vertices that could have a \emph{trivial minimum cut}, \ie a minimum cut, where one side contains only a single vertex. Figure~\ref{fig:local} depicts the contraction routines and Lemma~\ref{lem:local_crit} gives a more formal definition of them.

 \begin{lemma} \label{lem:local_crit}
    For an edge $e = (u,v) \in E$, $e$ is not part of any minimum cut, if $e$ fulfills at least one of the following criteria. Thus,  all minimum cuts of $G$ are still present in $G/e$ and $e$ can be contracted.
    \begin{enumerate}
      \item \prOne: $c(e) > \hat\lambda$
      \item \prTwo:
      \begin{itemize}
        \item $c(v) < 2 c(e)$ and $c(v) > \hat\lambda$, or
        \item $c(u) < 2 c(e)$ and $c(u) > \hat\lambda$
      \end{itemize}      
      \item \prThree:\\
        $\exists w \in V$ with 
      \begin{itemize}
        \item $c(v) < 2 \{c(v,w) + c(e)\}$ and $c(v) > \hat\lambda$, and
        \item $c(u) < 2 \{c(u,w) + c(e)\}$ and $c(u) > \hat\lambda$
      \end{itemize}
      \item \prFour: \\$c(e) + \sum_{w \in V} min\{c(v,w),c(u,w)\} > \hat\lambda$
    \end{enumerate}
 \end{lemma}

\begin{proof}
  \begin{enumerate}
    \item If $c(e) > \hat\lambda$, every cut that contains $e$ has capacity $> \hat\lambda$. Thus it can not be a minimal cut.
    \item Without loss of generality let $v$ be the vertex in question. The condition $c(v) < 2 c(e)$ means that $e$ is heavier than all other edges incident to $v$ combined. Thus, for any non-trivial cut that contains $e$, we can find a lighter cut by replacing $e$ with all other incident edges to $v$, \ie moving $v$ to the other side of the cut. As this is not true for the trivial minimum cut $(v,V \setminus v)$, we cannot contract an edge incident to a vertex that has weight $\leq \hat\lambda$.
    \item  This condition is similar to (2). Let there be a triangle ${u,v,w}$ in the graph in which it holds for both $u$ and $v$ that the two incident triangle edges are heavier than the sum of all other incident edges. Then, every cut that separates $u$ and $v$ can be improved by moving $u$ and $v$ into the same side. As the cut could have vertex $w$ on either side, both vertices need to fulfill this condition. To make sure that we do not contract any trivial minimum cut, we check that both $v$ and $u$ have weight $> \hat\lambda$ and thus can not represent a trivial minimum cut. 
    \item In this condition we check the whole shared neighborhood of vertices $u$ and $v$. Every cut that separates $u$ and $v$ must contain $e$ and for each shared neighbor $w$ at least one of the edges connecting them to $w$. Thus, we sum over the lighter edge connecting them to the shared neighbors and have a lower bound of the minimum cut that separates $u$ and $v$. If this is heavier than $\hat\lambda$, we know that no minimum cut separates $u$ and $v$.
  \end{enumerate}
  \end{proof}

The conditions \prOne{} and \prTwo{} can both be checked for the whole graph in a 
single run in linear time. While we can check condition \prThree{} when summing up the 
lighter incident edges for condition \prFour{}, exhaustively checking all triangles incurs 
a strictly worse than linear runtime, as a graph can have up to $\Theta(m^{3/2})$ 
triangles~\cite{schank2005finding}. Thus, we only perform linear-time runs as 
developed by Chekuri~\etal\cite{Chekuri:1997:ESM:314161.314315} by marking the 
neighborhood of $u$ and $v$ while we check the conditions and do not perform the test on
marked vertices.

\subsubsection{Vertices with one Neighbor}
\label{sss:degreeone}

Over the run of the algorithm, we occasionally encounter vertices that have only a single neighbor. Let $v$ be this vertex with one neighbor and $e = (v,w)$ be the only incident edge. As we update $\hat\lambda$ to the minimum degree whenever we perform a bulk edge contraction, $c(e) \geq \hat\lambda$: for an edge whose weight is $> \hat\lambda$, condition \prOne{} will contract it. For an edge whose weight is $\hat\lambda$, the edge represents a trivial minimum cut iff $\hat\lambda = \lambda$. This is the only minimum cut that contains $e$, as every non-trivial cut containing $e$ has higher weight. Thus, we can contract $e$ for now and remember that it was contracted. If $\hat\lambda$ is decreased, we can forget about these vertices as the cuts are not minimal. When we are finished, we can re-insert all contracted vertices that have a trivial minimum cut. We perform this reinsertion in a bottom-up fashion (\ie in reverse order to how they were contracted), as the neighbor $w$ could be contracted in a later contraction. 

\subsection{Finding all Minimum Cuts}
\label{ss:cactus}

We apply the reductions in the previous section exhaustively until they are not able to find a significant number of edges to contract. On the remaining graph we aim to find the cactus representation of all minimum cuts. Our algorithm for this purpose is based on the algorithm of Nagamochi, Nakao and Ibaraki~\cite{nagamochi2000fast}. While there is a multitude of algorithms for the problem of finding all minimum cuts, to the best of our knowledge there are no implementations accessible to the public and there is no practical experimentation on finding all minimum cuts. We base our algorithm on the algorithm of Nagamochi, Nakao and Ibaraki~\cite{nagamochi2000fast}, as their algorithm allows us to run the reduction routines previously detailed in between recursion steps.

We give a quick sketch of their algorithm, for details we refer the reader to \cite{nagamochi2000fast}. To find all minimum cuts in graph $G$, the algorithm chooses an edge $e=(s,t)$ in $G$ and uses a maximum flow $f$ to find the minimum s-t-cut $\lambda(s,t)$. If $\lambda(s,t) > \lambda$ there is no minimum cut that separates $s$ and $t$ and thus $e$ can be contracted. If $\lambda(s,t) = \lambda$, the edge is part of at least one minimum cut. They show that the strongly connected components $(V_1,\dots,V_k)$ of the residual graph $G_f$ represent all minimum cuts that contain $e$ (and potentially some more). For each connected component $V_i$, they build a graph $C_i$, in which all other connected components are contracted into a single vertex. We recurse on these component subgraphs and afterwards combine the minimum cut cactus graphs of the recursive calls to a cactus representation for $G$. 

The combination of the cactus graphs begins by building a cactus graph $C$ representing the set of strongly connected components, in which each $V_i$ is represented by a single vertex $v_i$. Each cactus $C_i$ is then merged with $C$ by replacing $v_i$ with $C_i$. For details we refer the reader to~\cite{nagamochi2000fast}.

As the contraction routines in Section~\ref{ss:contraction} usually mark a large amount of edges that can be contracted in bulk, we represent the graph in the compressed sparse row format\ifEnableExtend~\cite{tewarson1973sparse}\fi{}. This allows for fast and memory-efficient accesses to vertices and edges, however, we need to completely rebuild the graph in each bulk contraction and also keep vertex information about the whole graph hierarchy to be able to see which vertices in the original graph are encompassed in a vertex in a coarser vertex and to be able to re-introduce the cactus edges that were removed. While this is efficient for the bulk contractions performed in the previous section, in this section we often perform single-edge contractions or contract a small block of vertices. For fast running times these operations should not incur a complete rebuild of the graph data structure. We therefore use a mutable adjacency list data structure where each vertex is represented by a dynamic array of edges to neighboring vertices. Each edge stores its weight, target and the ID of its reverse edge (as we look at undirected graphs). This allows us to contract edges and small blocks in time corresponding to the sum of vertex degrees. For each vertex in the original graph, we store information which vertex currently encompasses it and every vertex keeps a list of currently encompassed vertices of the original graph. This information is updated during each edge contraction.
Inside this algorithm we re-run the contraction routines of Section~\ref{ss:contraction}. As they incur some computational cost and the graph does not change too much over different recursion steps, we only run the contraction routines every $10$ recursion steps.

\subsubsection{Edge Selection}

The recursive algorithm of Nagamochi, Nakao and Ibaraki~\cite{nagamochi2000fast} selects an arbitrary edge for the maximum flow problem in each recursion step. If this edge has connectivity equal to the minimum cut, we create a recursive subproblem for each connected component of the residual graph. In order to reduce the graph size - and thus the amount of work necessary - quickly, we aim to select edges in which the largest connected component of the residual graph is as small as possible. The edge selection strategy \texttt{Heavy} searches for the highest degree vertex $v$ and chooses the edge from $v$ to its highest degree neighbor. The strategy \texttt{WeightedHeavy} does the same, but uses the vertices whose weighted degree is highest. The idea is that an edge between high-degree vertices is most likely 'central' to the graph and thus manages to separate sizable chunks from the graph. The edge selection strategy \texttt{Central} aims to find a central edge more directly: we aim to find two vertices $u$ and $v$ with a high distance and take the central edge in their shortest paths. We find those vertices by performing a breadth-first search from a random vertex $w$, afterwards performing a breadth-first search from the vertex encountered last. We then take the central edge in the shortest path (as defined from the second breadth-first search) from the two vertices encountered last in the two breadth-first searches. The edge selection strategy \texttt{Random} picks a random edge.

\subsubsection{Additional Reductions}
\label{sss:degreetwo}

Over the course of this recursive contraction-based algorithm, we routinely encounter vertices with just two neighbors. Let $v$ be the vertex in question, which is connected to $u_0$ by edge $e_0$ and to $u_1$ by edge $e_1$. We look at four cases, each looking at whether the weight of $e_0$ being equal to the weight of $e_1$ and $c(v)$ being equal to $\lambda$, both conditions that can be checked in constant time.

$c(e_0) \neq c(e_1)$ and $c(v) > \lambda$: Without loss of generality let $e_0$ be the heavier edge. As $c(v) > \lambda$, the trivial cut $(\{v\}, V \setminus \{v\})$ is not a minimum cut. As by definition no cut in $G$ is smaller than $\lambda$, $\lambda(u_0, u_1) \geq \lambda$. Thus, excluding the path through $v$, they have a connectivity of $\geq \lambda - c(e_1)$ and any cut containing $e_0$ has weight $\geq \lambda - c(e_1) + c(e_0) > \lambda$ and can thus not be minimal. We therefore know that $e_0$ is not part of any minimum cuts and can be contracted according to Lemma~\ref{lem:contractible}.

$c(e_0) \neq c(e_1)$ and $c(v) = \lambda$: Without loss of generality let $e_0$ be the heavier edge. Analogously to the previous case we can show that no nontrivial cut contains $e_0$. In this case, where $c(v) = \lambda$, the trivial cut $(\{v\}, V \setminus \{v\})$ is minimal and therefore should be represented in the cactus graph. For all other minimum cuts that contain $e_1$, we know that $v$ and $u_0$ will be in the same block (as $c(e_0) > c(e_1)$). Thus, $v$ will be represented in the cactus as a leaf incident to $u_0$. We contract $e_0$ calling the resulting vertex $u^*$ and store which vertices of the original graph are represented by $v$. Then we recurse. On return from the recursion we check which cactus vertex now encompasses $u^*$ and add an edge from this vertex to a newly added vertex representing all vertices encompassed by $v$.

$c(e_0) = c(e_1)$ and $c(v) > \lambda$: in this case we are not able to contract any edges without further connectivity information.

$c(e_0) = c(e_1)$ and $c(v) = \lambda$: as $c(v) = \lambda$, the trivial cut $(\{v\}, V \setminus \{v\})$ is minimal. If there are other minimum cuts that contain either $e_0$ or $e_1$ (e.g. that separate $u_0$ and $u_1$), we know that by replacing $e_0$ with $e_1$ (or vice-versa) the cut remains minimal. Such a minimum cut exists iff $\lambda(u_0, u_1) = \lambda$. We contract $e_0$ and remember this decision. As $e_1$ is still in the graph (merged with $(u_0, u_1)$), we are able to find each cut that separates $u_0$ and $u_1$. If none exists, $\lambda(u_0, u_1) > \lambda$ and $u_0$ and $u_1$ will be contracted into a single vertex later in the algorithm. When leaving the recursion, we can thus re-introduce vertex $v$ as a leaf connected to the vertex encompassing $u_0$ and $u_1$. If $u_0$ and $u_1$ are in different vertices after leaving the recursion, there is at least one nontrivial cut that contains $e_1$. We thus re-introduce $v$ as a cycle vertex connected to $u_0$ and $u_1$, each with weight $\frac{\lambda}{2}$, and subtract $\frac{\lambda}{2}$ from $c(u_0, u_1)$.

In three out of the four cases presented here, we are able to contract an edge incident to a degree-two vertex. We can check these conditions in total time $\Oh{n}$ for the whole graph. Over the course of the algorithm, we perform edge contractions and thus routinely encounter vertices whose neighborhood has been contracted and thus have a degree of two. Thus, these reductions are able to reduce the size of the graph significantly even if the initial graph is rather dense and does not have a lot of low degree vertices.

\subsection{Putting it all together}
\label{ss:combine}

\begin{algorithm}[t]
  \caption{Algorithm to find all minimum cuts}\label{alg:overview}
  \begin{algorithmic}[1]
  \Procedure{FindAllMincuts}{$G = (V,E)$}
  \State $\hat\lambda \gets \text{VieCut}(G)$ \cite{henzinger2018practical} 
  \While{not converged}
  \State $(G,D_1, \hat\lambda) \gets \text{contract degree-one vertices}(G, \hat\lambda)$
  \State $(G, \hat\lambda) \gets \text{connectivity-based contraction}(G, \hat\lambda)$
  \State $(G, \hat\lambda) \gets \text{local contraction}(G, \hat\lambda)$
  \EndWhile
  \State $\lambda \gets \text{FindMinimumCutValue}(G)$
  \State $C \gets \text{RecursiveAllMincuts}(G, \lambda)$ (\cite{nagamochi2000fast})
  \State $C \gets \text{reinsert vertices}(C, D_1)$
  \Return $(C, \lambda)$
  \EndProcedure
  \end{algorithmic}
  \end{algorithm}

Algorithm~\ref{alg:overview} gives an overview over our algorithm to find all minimum cuts. Over the course of the algorithm we keep an upper bound $\hat\lambda$ for the minimum cut, initially set to the result of the inexact variant of the \texttt{VieCut} minimum cut algorithm~\cite{henzinger2018practical}. While the \texttt{VieCut} algorithm also offers an exact version~\cite{henzinger2019shared}, we use the inexact version, as it is considerably faster and gives a low upper bound for the minimum cut, usually equal to the minimum cut. As described in Section~\ref{ss:contraction}, we use this bound to contract degree-one vertices, high-connectivity edges and edges whose local neighborhood guarantees that they are not part of any minimum cut. We repeat this process until it is converged, as an edge contraction can cause other edges in the neighborhood to also become safely contractible. 
As this process often incurs a long tail of single edge contractions, we stop if the number of vertices was decreased by less than $1\%$ over a run of all contraction routines. 

We then use the minimum cut algorithm of Nagamochi, Ono and Ibaraki~\cite{nagamochi1992computing,nagamochi1994implementing} on the remaining graph, as the following steps need the correct minimum cut. To find all minimum cuts in the contracted graph, we call our optimized version of the algorithm of Nagamochi~\etal\cite{nagamochi2000fast}, as sketched in Section~\ref{ss:cactus}, and afterwards re-insert all minimum cut edges that were previously deleted. Before each recursive call of the algorithm of Nagamochi~\etal\cite{nagamochi2000fast}, we contract edges incident to degree-one and eligible degree-two vertices. Every $10$ recursion levels we additionally check for connectivity-based edge contractions and local contractions. 

\subsection{Shared-Memory Parallelism}
Algorithm~\ref{alg:overview} employs shared-memory parallelism in every step. When we run the algorithm in parallel, we use the parallel variant of VieCut~\cite{henzinger2018practical}. Local contraction and marking of degree one vertices are parallelized using OpenMP~\cite{dagum1998openmp}. For the first round of connectivity-based contraction, we use the parallel connectivity certificate used in the shared-memory parallel minimum cut algorithm by Henzinger~\etal\cite{henzinger2019shared}. This connectivity certificate is essentially a parallel version of the connectivity certificate of Nagamochi~\etal\cite{nagamochi1992computing,nagamochi1994implementing}, in which the processors divide the work of computing the connectivity bounds for all edges of the graph. In subsequent iterations every processor runs an independent run of the connectivity certificate of Nagamochi~\etal on the whole graph starting from different random vertices in the graph. As the connectivity bounds given by the algorithm heavily depend on the starting vertex, this allows us to find significantly more contractible edges per round than running the connectivity certificate only once. 

We use the shared-memory parallel minimum cut algorithm of Henzinger~\etal\cite{henzinger2019shared} to find the exact minimum cut of the graph. The algorithm of Nagamochi~\etal\cite{nagamochi2000fast} is not shared-memory parallel, however we usually manage to contract the graph to a size proportional to the minimum cut cactus before calling them. Unfortunately it is not beneficial to perform the recursive calls embarrassingly parallel, as in almost all cases one of the connected components of the residual graph contains the vast majority of vertices and thus also has the overwhelming majority of work. 

\section{Applications} \label{s:balanced}

We can use the minimum cut cactus $C_G$ to find a minimum cut fulfilling certain balance criteria, such as a most balanced minimum cut, e.g. a minimum cut $(A,V \setminus A)$ that maximizes min$(|A|,|V \setminus A|)$. Note that this is not equal to the most balanced $s$-$t$-cut problem, which is NP hard~\cite{bonsma2007most}. Following that we show how to modify the algorithm to find the optimal minimum cut for other optimization functions.

One can find a most balanced minimum cut trivially in time $\Oh{(n^*)^3}$, as one can enumerate all $\Oh{(n^*)^2}$ minimum cuts~\cite{karger2000minimum} and add up the number of vertices of the original graph $G$ on either side. We now show how to find a most balanced minimum cut of a graph $G$ in $\Oh{n^* + m^*}$ time, given the minimum cut cactus graph $C_G$. 

For every cut $(A, V\backslash A)$, we define the balance $b(A)$ (or $b(V\backslash A))$ of the cut as the number of vertices of the original graph encompassed in the lighter side of the cut. Recall that for any node $v \in V_G$, $c(v)$ is the number of vertices of $G$ represented by $v$. For a leaf $v \in V_G$, we set its weight $w(v) = c(v)$ and set the balance $b(v)$ to be the minimum of $w(v)$ and $n - w(v)$. We root $C_G$ in an arbitrary vertex and depending on that root define $w(v)$ as the sum of vertex weights in the subcactus rooted in $v$; and $b(v)$ accordingly. For a cycle $C = \{c_1, \dots, c_i\}$, we define $b(c_j,\dots,c_{k \mod i})$ with $0 \geq j \geq k$ analogously as the balance of the minimum cut splitting the cycle so that the sub-cacti rooted in $c_j, \dots, c_{k \mod i}$ are on one side of the cut and the rest are on the other side (see blue line in Figure~\ref{fig:mostbalanced} for an example).

\begin{figure}[t!]\centering
  \includegraphics[width=.5\textwidth]{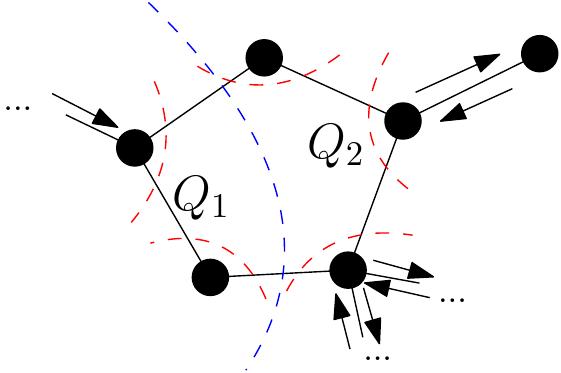}
  \caption{Cycle check in balanced cut algorithm\label{fig:mostbalanced}}
\end{figure}

Let $T_G$ be the tree representation of $C_G$ where each cycle in $C_G$ is contracted into a single vertex. We perform a depth-first search on $T_G$ rooted on an arbitrary vertex and check the balance of every cut in $T_G$ when backtracking. 

As $C_G$ is not necessarily a tree, we might encounter cycles and we explain next how to extend the depth first search to handle such cycles. Let $\mathcal{C} = \{c_0,\dots,c_{i-1}\}$ be a cycle and $c_0$ be the vertex encountered first by the DFS. Due to the cactus graph structure of $C_G$, the depth-first search backtracks from a vertex $v_{cy}$ in $T_G$ that represents $\mathcal{C}$ only after all subtrees rooted in $\mathcal{C}$ are explored. Thus, we know the weight of all subtrees rooted in vertices $c_1,\dots,c_{i-1}$ when backtracking. The weight of $c_0$ is equal to $n$ minus the sum of these sub-cactus weights. 

Examining all cuts in the cycle would take $i^3$ time, but as we only want to find the most balanced cut, we can check only a subset of them, as shown in Algorithm~\ref{alg:cycle}. $Q_1$ and $Q_2$ are \emph{queues}, thus elements are ordered and the following operations are supported: \emph{queue} adds an element to the back of the queue, called the \emph{tail} of the queue, \emph{dequeue} removes the element at the front of the queue, called the \emph{head} of the queue. We implicitly use the fact that queues can only be appended to, thus an element $q$ was added to the queue after all elements that are closer to the head of the queue and before all elements that are closer to its tail.

\begin{algorithm}
  \caption{Algorithm to find most balanced cut in cycle $\{c_0,\dots,c_{i-1}\}$}\label{alg:cycle}
  \begin{algorithmic}[1]
  \Procedure{BalanceInCycle}{$G = (V,E), C =\{c_1,\dots,c_i\}$}
  \State $b_{OPT} \gets 0$
  \State $Q_1 =$ Queue($\{\}$)
  \State $Q_2 =$ Queue($\{c_0,c_1,\dots,c_{i-1}\}$)
  \While{$c_0$ not $Q_1$.head() for second time}
    \State $b_{OPT} \gets$ checkBalance($Q_1, Q_2$)
    \If{$w(Q_1) > w(Q_2)$} 
      \State $Q_2$.queue($Q_1$.dequeue())
    \Else
      \State $Q_1$.queue($Q_2$.dequeue())
    \EndIf
  \EndWhile
  \Return $b_{OPT}$
  \EndProcedure
  \end{algorithmic}
  \end{algorithm}

The weight of a queue $w(Q)$ is denoted as the weight of its contents. For queue $Q = \{c_{j \bmod i},\dots,c_{k \bmod i}\}$ with $0 \leq j \leq k$, we use the notation $w_{j \bmod i,k \bmod i}$ to denote the weight of $Q$ and $\overline{w_{j \bmod i, k \bmod i}}$ as the weight of the queue that contains all cycle vertices not in $Q$. 

In every step of the algorithm, the cut represented by the current state of the queues consists of the two edges connecting the queue heads to the tails of the respective other queue. Initially $Q_1$ is empty and $Q_2$ contains all elements, in order from $c_0$ to $c_{i-1}$. In every step of the algorithm, we dequeue one element and queue it in the other queue. Thus, at every step each cycle vertex is in exactly one queue. When we check the balance of a cut, we compute the weight of each queue at the current point in time; and update $b_{OPT}$, the best balance found so far, if $(Q_1, Q_2)$ is more balanced. As we only move one cycle vertex in each step, we can check the balance of an adjacent cut in constant time by adding and subtracting the weight of the moved vertex to the weights of each set.

\begin{lemma}
  Algorithm~\ref{alg:cycle} terminates after $O(i)$ steps.
\end{lemma}

\begin{proof}
  In each step of Algorithm~\ref{alg:cycle}, one queue head is moved to the other queue. The algorithm terminates when $c_0$ is the head of $Q_1$ for the second time. In the first step, $c_0$ is moved to $Q_1$, as the empty queue $Q_1$ is the lighter one. The algorithm terminates after $c_0$ then performs a full round through both queues and is the head of $Q_1$ again. At termination, $c_0$ was thus moved a total of three times, twice from $Q_2$ to $Q_1$ and once the other way. As no element can 'overtake' $c_0$ in the queues, every vertex will be moved at most three times. Thus, we enter the loop at most $3i$ times, each time only using a constant~amount~of~time.
\end{proof}

In Algorithm~\ref{alg:cycle}, we only check the balance of a subset of cuts represented by edges in the cycle $C$. Lemma~\ref{lem:mostbal} shows that none of the disregarded cuts can have balance better than $b_{OPT}$ and we thus find the most balanced minimum cut. We call a cut disregarded if its balance was never checked (Line $6$), and considered otherwise. In order to prove correctness of Algorithm~\ref{alg:cycle}, we first show the following Lemma:

\begin{lemma}\label{lem:head}
Each vertex in the cycle is dequeued from $Q_1$ at least once in the algorithm.
\end{lemma}

\begin{proof}
  The algorithm terminates when $c_0$ is the head of $Q_1$ for the second time. For this, it needs to be moved from $Q_2$ to $Q_1$ twice. As we queue elements to the back of a queue, all vertices are dequeued from $Q_2$ before $c_0$ is dequeued from it for the second time. In order for $c_0$ to become the head of $Q_1$ again, all elements that were added beforehand need to be dequeued from $Q_1$.
\end{proof}

\begin{lemma}\label{lem:mostbal}
  Algorithm~\ref{alg:cycle} finds the most balanced minimum cut represented by cycle~$C$.
\end{lemma}

\begin{proof}
  We now prove for each $c_l \in \mathcal{C}$ that all disregarded cuts containing the cycle edge separating $c_l$ from $c_{(l-1) \bmod i}$ are not more balanced than the most balanced cut found so far. As no disregarded cut can be more balanced than the most balanced cut considered in the algorithm, the output of the algorithm is the most balanced minimum cut; or one of them if multiple cuts of equal balance exist.

  \begin{figure}[H] \centering
    \includegraphics[width=.45\textwidth]{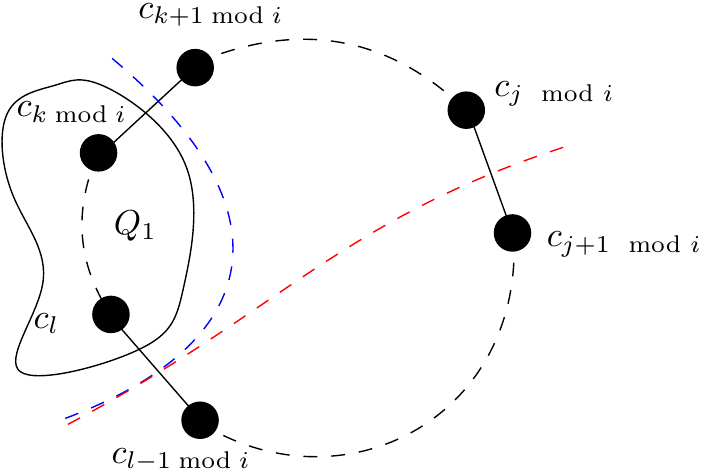}
    \caption{State of $Q_1$ at time $t_l$ (cut in blue). Cut in red denotes cut considered at time $t^*$\label{fig:cycle_slice}}.
    \vspace*{-.5cm}
  \end{figure}

  Let $t_l$ be the time that $c_l$ becomes the head of $Q_1$ for the first time. Figure~\ref{fig:cycle_slice} shows the state of $Q_1$ at that point in time. Let $c_{k \bmod i}$ be the tail of $Q_1$ at time $t_l$ for some integer $k$. Right before $t_l$, $c_{l-1 \bmod i}$ was head of the heavier queue $Q_1$ and thus dequeued, \ie $Q_1 = \{c_{l-1 \bmod i},\dots,c_{k \bmod i}\}$ has weight $w_{l-1 \bmod i, k \bmod i} \geq \overline{w_{l-1 \bmod i, k \bmod i}}$ and $c_l$ is now head of $Q_1$.

  From this point $t_l$ the algorithm considers cuts that separate $c_l$ from $c_{l-1 \bmod i}$. While $Q_1$ is not heavier than $Q_2$, we add more elements to the tail of $Q_1$ (and check the respective cuts) until $Q_1$ is the heavier queue. Let $t^*$ be the time when this happens and $c_{j \bmod i}$ with $j \geq k$ be the tail of $Q_1$ at this point. Note that at time $t^*$, $c_l$ is about to be dequeued from $Q_1$. The red cut in Figure~\ref{fig:cycle_slice} shows the cut at time $t^*$, where $w_{c_l,c_{j \bmod i}} > \overline{w_{c_l,c_{j \bmod i}}}$.
  
  We now prove that all cuts in which $c_l$ is the head of $Q_1$ and its tail is not between $c_{k \bmod i}$ and $c_{j \bmod i}$ cannot be more balanced than the most balanced cut considered so far.
  
  For all cuts where $c_l$ is head of $Q_1$ and $Q_1$ also contains $c_{j+1 \bmod i}$, $Q_1$ is heavier than $w_{l,j \bmod i}$, as it contains all elements in $c_l,\dots,c_{j \bmod i}$ plus at least one more. As $w_{l,j \bmod i} > \overline{w_{l,j \bmod i}}$, \ie $Q_1$ is already heavier when $c_{j \bmod i}$ is its tail, all of these cuts are less balanced than $(\{c_l,\dots,c_{j \bmod i}\}, \mathcal{C} \backslash \{c_l,\dots,c_{j \bmod i}\})$.

  For the cuts in which $c_{k \bmod i}$ is in $Q_2$, \ie $Q_1$ is lighter than at time $t_l$, we need to distinguish two cases, depending on whether $w_{l,k \bmod i}$ is larger than $\overline{w_{l,k \bmod i}}$ or not.

  If $w_{l,k \bmod i} \leq \overline{w_{l,k \bmod i}}$, all cuts in which $c_l$ is the head of $Q_1$ and $c_{k \bmod i}$ is in $Q_2$ are less balanced than $(\{c_l,\dots,c_{k \bmod i}\}, \mathcal{C} \backslash \{c_l,\dots,c_{k \bmod i}\})$, as $Q_1$ is lighter than it is at $t_l$, where it was already not the heavier queue.

  If $w_{l,k \bmod i} > \overline{w_{l,k \bmod i}}$, there might be cuts in which $c_l$ is the head of $Q_1$ that are more balanced than $(\{c_l,\dots,c_{k \bmod i}\}, \mathcal{C} \backslash \{c_l,\dots,c_{k \bmod i}\})$ in which $Q_1$ is lighter than at time $t_l$. Thus, consider time $t'$ when $c_{k \bmod i}$ was added to $Q_1$. Such a time must exist, since $Q_1$ is initially empty. As $c_{k \bmod i}$ is already the tail of $Q_1$ at time $t_l$, $t' < t_l$. At that time $Q_1$ contained $c_{l-1 \bmod i}, \dots, c_{k-1 \bmod i}$ and potentially more vertices. 
  
  Still, $w_{l-1 \bmod i, k-1 \bmod i} \leq \overline{w_{l-1 \bmod i, k-1 \bmod i}}$, as otherwise $c_{k \bmod i}$ would not have been added to $Q_1$. Obviously $w_{l-1 \bmod i, k-1 \bmod i} > w_{l, k-1 \bmod i}$, as $Q_1$ is even lighter when $c_{l-1 \bmod i}$ is dequeued. As $w_{l-1 \bmod i, k-1 \bmod i}$ is already not heavier than its complement, $(\{c_l,\dots,c_{k-1 \bmod i}\}, \mathcal{C} \backslash \{c_l,\dots,c_{k-1 \bmod i}\})$ is more imbalanced than the cut examined just before time $t'$. Thus, all cuts where $c_l$ is the head of $Q_1$ and $c_{k-1 \bmod i}$ is in $Q_2$ are even more imbalanced, as $Q_1$ is even lighter.

  Coming back to the outline shown in Figure~\ref{fig:cycle_slice}, we showed that for all cuts in which $c_l$ is head of $Q_1$ and $Q_1$ is lighter than at time $t_l$ (left of blue cut) and all cuts where $Q_1$ is heavier than at time $t^*$ (below red cut) can be safely disregarded, as a more balanced cut than any of them was considered at some point between $t'$ and $t^*$. The algorithm considers next all cuts with $c_l$ as head of $Q_1$ and the tail of $Q_1$ between $c_{k \bmod i}$ and $c_{j \bmod i}$. Thus, the algorithm will return a cut that is at least as balanced as the most balanced cut that separates $c_l$ and $c_{l-1 \bmod i}$. This is true for every cycle vertex $v_l \in \mathcal{C}$, which concludes the proof.
\end{proof}

This allows us to perform the depth-first search and find the most balanced minimum cut in $C_G$ in time $\Oh{n^* + m^*}$. This algorithm can be adapted to find the minimum cut of any other optimization function of a cut that only depends on the (weight of the) edges on the cut and the (weight of the) vertices on either side of the cut. In order to retain the linear running time of the algorithm, the function needs to be evaluable in constant time on a neighboring cut. For example, we can find the minimum cut of lowest conductance. The conductance of a cut $(S, V \setminus S)$ is defined as $\frac{\lambda(S,(V\setminus S))}{min(a(S), a(V\setminus S))}$, where $a(S)$ is the sum of degrees for all vertices in set $S$. Note that this is not the minimum conductance cut problem, which is NP-hard~\cite{andersen2008algorithm}, as we only look at the minimum cuts. To find the minimum cut of lowest conductance, we set the weight of a vertex $v_{C_G} \in C_G$ to the sum of vertex degrees encompassed in $v_{C_G}$. Otherwise the algorithm remains the same.

\section{Experiments and Results} \label{s:experiments}

We now perform an experimental evaluation of the proposed algorithms.
This is done in the following order: first analyze the impact of algorithmic components on our minimum cut algorithm in a non-parallel setting, i.e.~we compare different variants for edge selection and see the impact of the various optimizations detailed in this work. Afterwards, we report parallel speedup on a variety of large graphs.

\paragraph*{Experimental Setup and Methodology}

We implemented the algorithms using \CC-17 and compiled all code using g++ version 8.3.0 with full optimization (\texttt{-O3}). Our experiments are conducted on a machine with two Intel Xeon Gold 6130 processors with 2.1GHz with 16 CPU cores each and $256$ GB RAM in total. We perform five repetitions per instance and report average running time. In this section we first describe our experimental methodology. Afterwards, we evaluate different algorithmic choices in our algorithm and then we compare our algorithm to the state of the art. When we report a mean result we give the geometric mean as problems differ significantly in cut size and time. Our code is freely available under the permissive MIT license\ifDoubleBlind{\footnote{Link removed, as this submission is double-blind}.}\else{\footnote{\url{https://github.com/alexnoe/VieCut}}\fi{}.

\paragraph*{Instances}

We use a variety of graphs from the 10th DIMACS Implementation challenge~\cite{bader2013graph} and the SuiteSparse Matrix Collection~\cite{davis2011university}. These are social graphs, web graphs, co-purchase matrices, cooperation networks and some generated instances. Table~\ref{tab:small} shows a set of smaller instances and Table~\ref{tab:large} shows a set of larger and harder to solve instances. All instances are undirected. If the original graph is directed, we generate an undirected graph by removing edge directions and then removing duplicate edges. If a network has multiple connected components, we run on the largest.

As most large real-world networks have cuts of size 1, finding all minimum cuts becomes essentially the same as finding all bridges, which can be solved in linear time using depth-first search~\cite{tarjan1972depth}. However, usually there is usually one huge block that is connected by minimum cuts to a set of small and medium size blocks. Thus, we use our algorithm to generate a more balanced set of instances. We find all minimum cuts and contract each edge that does not connect two vertices of the largest block. Thus, the remaining graph only contains the huge block and is guaranteed to have a minimum cut value $>\lambda$. We use this method to generate multiple graphs with different minimum cuts for each instance.

\subsection{Edge Selection}

Figure~\ref{fig:edgeselect} shows the results for the graphs in Table~\ref{tab:small}. We compute the cactus graph representing all minimum cuts using the edge selection variants \texttt{Random}, \texttt{Central}, \texttt{Heavy} and \texttt{HeavyWeighted}, as detailed in Section~\ref{ss:cactus}. As we want a majority of the running time in the recursive algorithm of Nagamochi~\etal\cite{nagamochi2000fast}, where we actually select edges, we run a variant of our algorithm that only contracts edges using connectivity-based contraction and then runs the algorithm of Nagamochi~\etal\cite{nagamochi2000fast}.

\begin{figure}[t!]\centering
  \includegraphics[width=.65\textwidth]{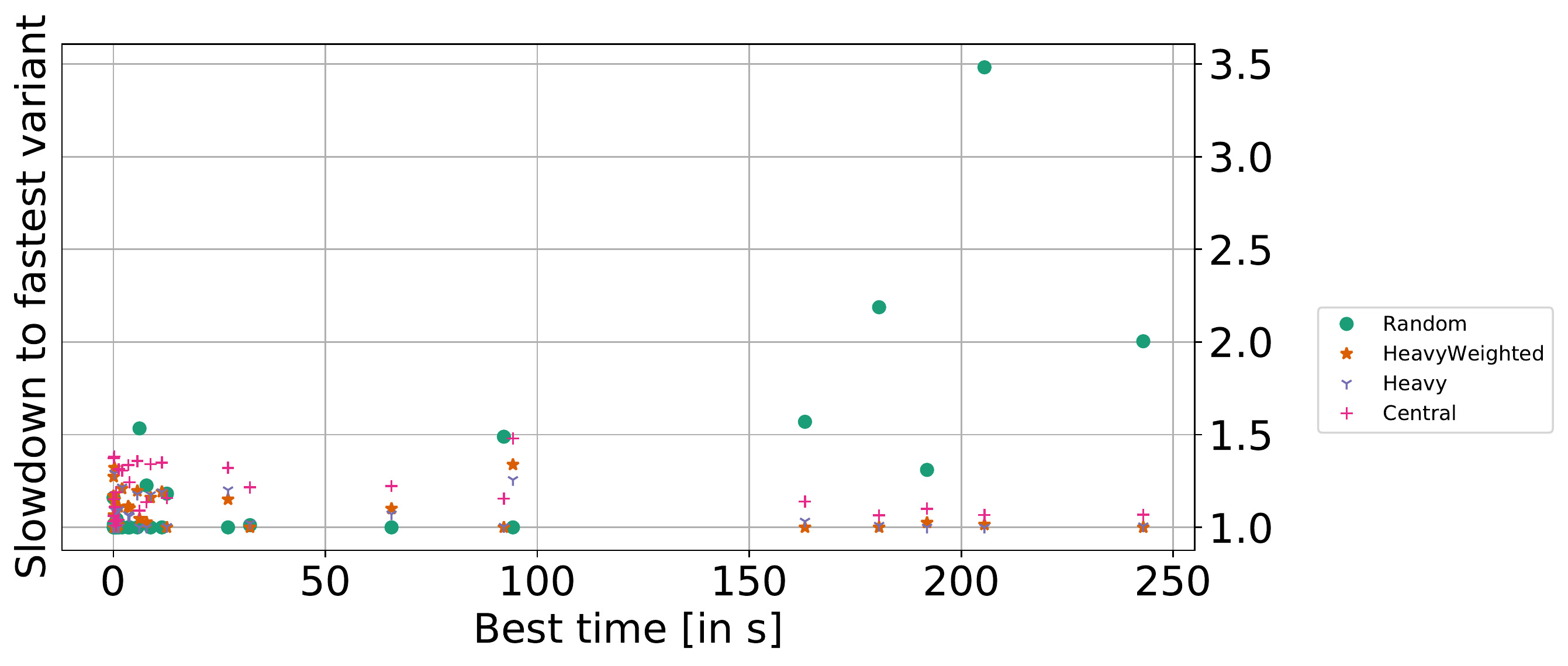}
  \caption{Effect of edge selection strategies.\label{fig:edgeselect}}
\end{figure}%

We can see that in the graphs which cannot be contracted quickly, \texttt{Random} is significantly slower than all other variants. On \texttt{cnr-2000}, \texttt{Random} takes over $700$ seconds in average, whereas all other variants finish in approximately $200$ seconds. This happens independently of the random seed used, there is no large deviation in the running time on any of the graphs. On almost all graphs, the variants \texttt{Heavy} and \texttt{HeavyWeighted} are within $3\%$ of each other, which is not surprising, as the variants are almost identical. While it optimizes for 'edge centrality' very directly, \texttt{Central} has two iterations of breadth-first search in each edge selection and thus a sizable overhead. For this reason it is usually $5-15\%$ slower than \texttt{Heavy} and is not the fastest algorithm on any graph. On graphs with large $n^*$, all three variants manage to shrink the graph significantly faster than \texttt{Random}.

On graphs with a low value of $n^*$, we can see that \texttt{Random} is slightly faster than the other variants. There is no significant difference in the shrinking of the graph, as almost all selected edges have connectivity larger than $\lambda$ and thus only trigger a single edge contraction anyway. Thus, not spending the extra work of finding a `good' edge results in a slightly lower running time.
In the following we will use variant \texttt{Heavy}, which is the only variant that is never more than $30\%$ slower than the fastest variant on any graph.

\subsection{Optimization}

We now examine the effect of the different optimizations. For this purpose, we benchmarks different variants on a variety of graphs. We hereby compare the following variants that build on one another: as a baseline, \optZero{} runs the algorithm of Nagamochi, Nakao and Ibaraki~\cite{nagamochi2000fast} on the input graph. \optOne{} additionally runs VieCut~\cite{henzinger2018practical} to find an upper bound for the minimum cut and uses this to contract high-connectivity edges as described in Section~\ref{sss:connectivity}. In addition to this, \optTwo{} also contracts edges whose neighborhood guarantees that they are not part of any minimum cut, as described in Section~\ref{sss:local} and Lemma~\ref{lem:local_crit}. \optThree{} runs also the last remaining contraction routine from Algorithm~\ref{alg:overview}, contraction and re-insertion of degree-one vertices as described in Section~\ref{sss:degreeone}.

\begin{figure}[t]\centering
  \includegraphics[width=.65\textwidth]{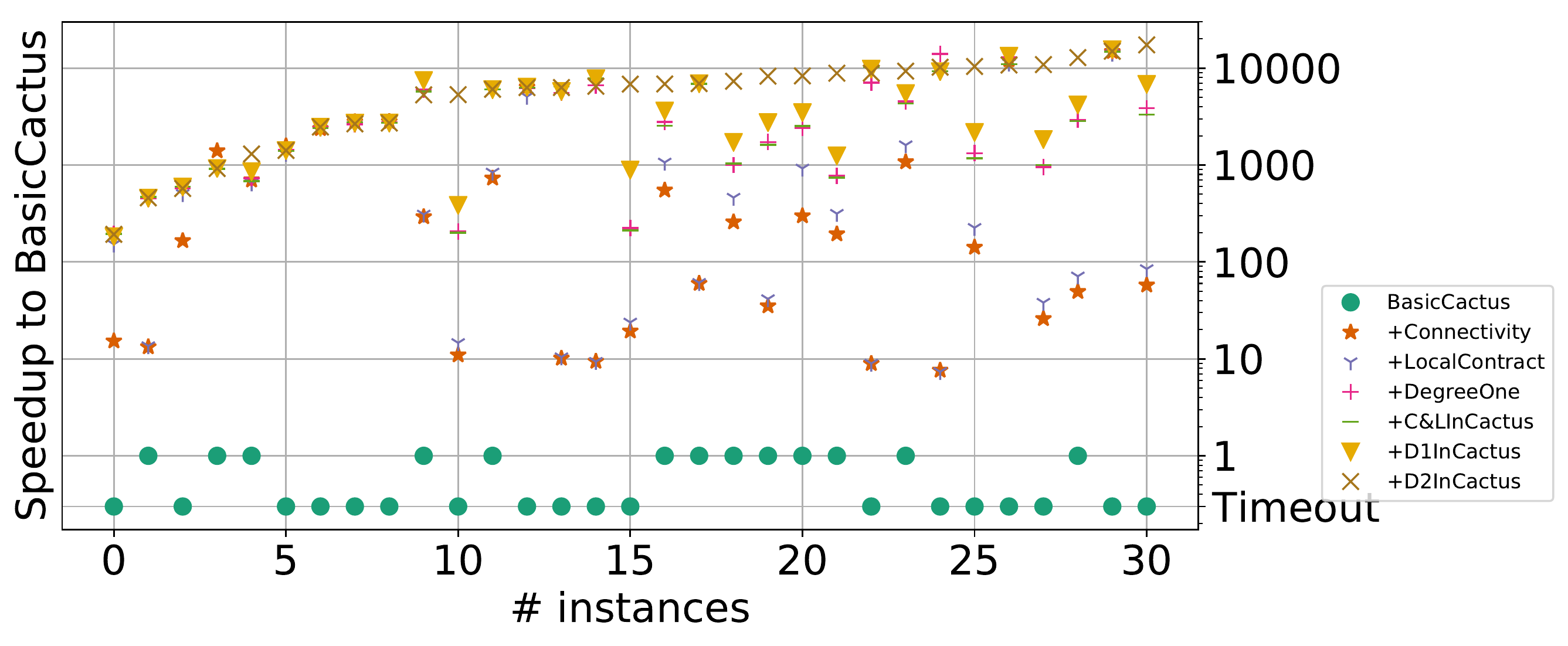}
  \caption{Speedup to \texttt{BasicCactus} on small graphs (Table~\ref{tab:small}).\label{fig:smallopt}}
\end{figure}

\begin{figure}[t]\centering
  \includegraphics[width=.65\textwidth]{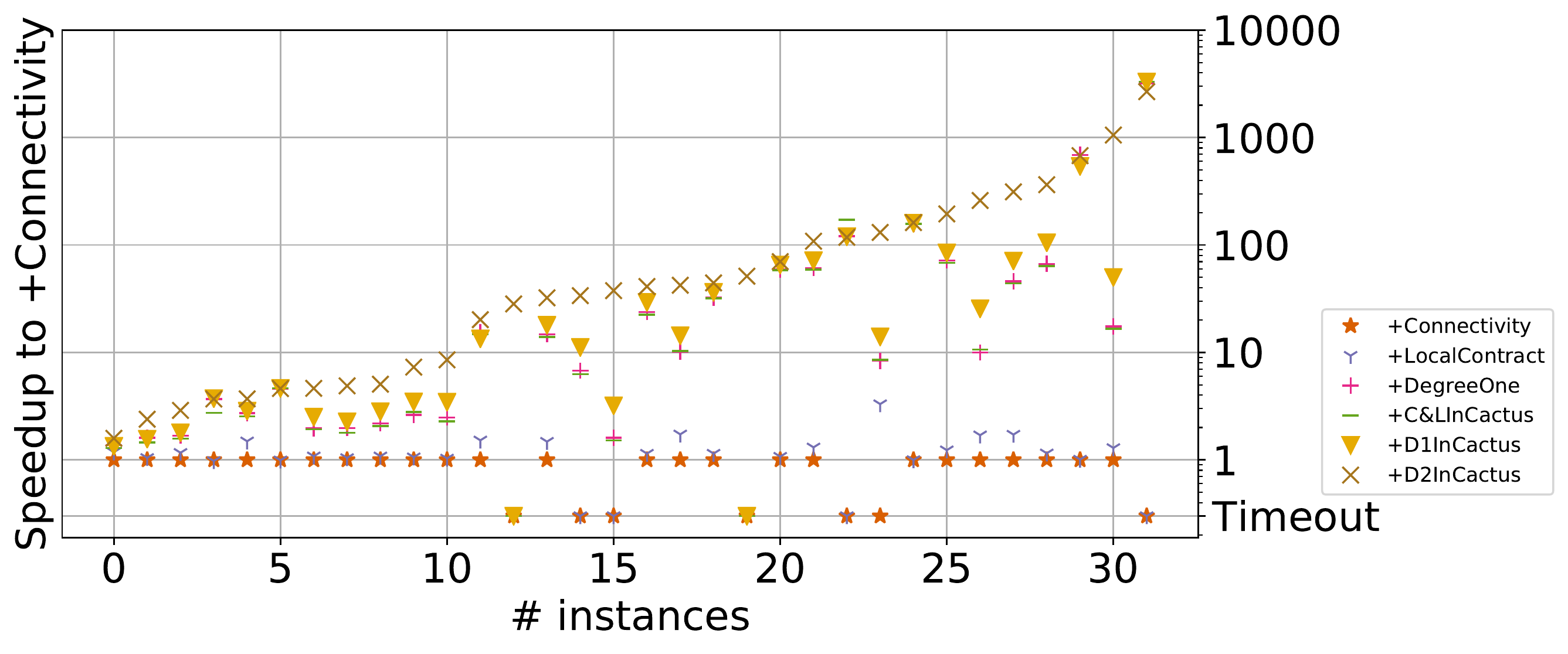}
  \caption{Speedup to \texttt{+Connectivity} on large graphs (Table~\ref{tab:large}).\label{fig:largeopt}}
\end{figure}

\optFour{} additionally runs high-connectivity and local contraction in every tenth recursion step. \optFive{} additionally contracts and re-inserts degree-one vertices in every recursion step. \optSix{} also runs the degree-two contraction as described in Section~\ref{sss:degreetwo}. We compare these variants on the graphs in Tables~\ref{tab:small}~and~\ref{tab:large}. We use a timeout of $30$ minutes on these graphs. If the baseline algorithm does not finish in the allotted time, we report speedup to the timeout, so a lower bound for the actual speedup.

Figure~\ref{fig:smallopt} shows the speedup of all variants to the baseline variant \optZero{} on all graphs in Table~\ref{tab:small}. We can see that already just adding \optOne{} gives a speedup of more than an order of magnitude for each of the graphs in the dataset. Most of the other optimizations manage to improve the running time of at least some instances by a large margin. Especially \optThree{}, which is the first contraction for edges that are in a minimum cut, has speedups of multiple orders of magnitude in some instances. This is the case as minimum cut edges that are incident to a degree-one vertex previously incur a flow problem on the whole graph each. However, it is very easy to see that the edge will be part of exactly one minimum cut, thus we can contract and re-insert it in constant time. Especially in graphs whose minimum cut is $1$, all edges can be quickly contracted, as they will either be incident to a degree-one vertex or be quickly certified to have a connectivity value of $>1$.

While rerunning \texttt{Connectivity} and \texttt{LocalContract} inside of the recursive algorithm of Nagamochi~\etal\cite{nagamochi2000fast} does usually not yield a large speedup, many graphs develop degree-one vertices by having their whole neighborhood contracted. Thus, \optFive{} has a significant speedup for most graphs in which $n^*$ is sufficiently large. \optSix{} has an even larger speedup on these graphs, even when the minimum cut is significantly higher than $2$, as there are often cascading effects where the contraction of an edge incident to a degree-two vertex often lowers the degree of neighboring vertices to two.

Figure~\ref{fig:largeopt} shows the speedup of all variants. As variant \optZero{} is not able to solve any of these instances in $30$ minutes, we use \optOne{} as a baseline. The results are similar to Figure~\ref{fig:smallopt}, but we can see even clearer how useful the contraction of degree-two vertices is in finding all minimum cuts: \optSix{} often has a speedup of more than an order of magnitude to all other variants and is the only variant that never times out.

\begin{table}[t!]
  \centering
  \caption{Huge social and web graphs. $n^*$ denotes number of vertices in cactus graph, max $n$ and max $m$ denote size of smaller block in most balanced cut\label{tab:huge}}
  \resizebox{\textwidth}{!}{
  \begin{tabular}{|r|r|r|r|r|r|r|r|r|}
    \hline
    Name & $n$ & $m$ & $n^*$ & max. $n$ & max. $m$ & seq. t & par. t\\ \hline\hline
    friendster & $65.6$M & $1.81$B & $13.99$M & \numprint{897} & \numprint{1793} & $1266.35$s & $138.34$s \\\hline
    twitter7 & $41.7$M & $1.20$B & $1.93$M & \numprint{47} & \numprint{1893} & $524.86$s & $72.51$s \\\hline
    uk-2007-05 & $104.3$M & $3.29$B & $9.66$M & \numprint{49984} & $13.8$M & $229.18$s & $40.16$s\\\hline
  \end{tabular}
  }
\end{table}

\subsection{Shared-memory Parallelism}

Table~\ref{tab:huge} shows the average running times of our algorithm both sequential and with $16$ threads on huge social and web graphs. Each of these graphs has more than a billion of edges and more than a million vertices in the cactus graph depicting all minimum cuts. On these graphs we have a parallel speedup factor of $5.7$x to $9.1$x using $16$ threads. On all of these graphs, a large part of the running time is spent in the first iteration of the kernelization routines, which already manages to contract most dense blocks in the graph. Thus, all subsequent operations can be performed on significantly smaller problems and are therefore much faster.  

\vspace*{-.1cm}
\section{Conclusion}

We engineered an algorithm to find all minimum cuts in large undirected graphs. Our algorithm combines multiple kernelization routines with an engineered version of the algorithm of Nagamochi, Nakao and Ibaraki~\cite{nagamochi2000fast} to find all minimum cuts of the reduced graph. Our experiments show that our algorithm can find all minimum cuts of huge social networks with up to billions of edges and millions of minimum cuts in a few minutes on shared memory. We found that especially the contraction of high-connectivity edges and efficient handling of low-degree vertices can give huge speedups. \ifEnableExtend Additionally we give a linear time algorithm to find the most balanced minimum cut given the cactus graph representation of all minimum cuts.\fi{} Future work includes finding near-minimum cuts.

\vspace*{-.1cm}

\bibliographystyle{plainurl}
\bibliography{paper}

\newpage 
\vfill

\begin{appendix}
\section{Graph Instances}

\begin{table}[H]
  \caption{Set of small graphs from various sources\label{tab:small}}
  \begin{tabular}{|r|r|r|r|r|}
    \hline
    Name & $n$ & $m$ & $\lambda$ & $n^*$ \\
    \hline\hline
    amazon & \numprint{64813} & \numprint{153973} & \numprint{1} & \numprint{10068} \\\hline
    auto & \numprint{448695} & $3.31M$ & \numprint{4} & \numprint{43} \\
    & \numprint{448529} & $3.31M$ & \numprint{5} & \numprint{102} \\
    & \numprint{448037} & $3.31M$ & \numprint{6} & \numprint{557} \\
    & \numprint{444947} & $3.29M$ & \numprint{7} & \numprint{1128} \\
    & \numprint{437975} & $3.24M$ & \numprint{8} & \numprint{2792} \\
    & \numprint{418547} & $3.10M$ & \numprint{9} & \numprint{5814} \\ \hline
    caidaRouterLevel & \numprint{190914} & \numprint{607610} & \numprint{1} & \numprint{49940} \\\hline
    cfd2& \numprint{123440} & $1.48M$ & \numprint{7} & \numprint{15} \\\hline
    citationCiteseer & \numprint{268495} & $1.16M$ & \numprint{1} & \numprint{43031} \\
     & \numprint{223587} & $1.11M$ & \numprint{2} & \numprint{33423} \\
     & \numprint{162464} & \numprint{862237} & \numprint{3} & \numprint{23373} \\
     & \numprint{109522} & \numprint{435571} & \numprint{4} & \numprint{16670} \\
     & \numprint{73595} & \numprint{225089} & \numprint{5} & \numprint{11878} \\
     & \numprint{50145} & \numprint{125580} & \numprint{6} & \numprint{8770} \\\hline
    cnr-2000 & \numprint{325557} & $2.74M$ & \numprint{1} & \numprint{87720} \\
    & \numprint{192573} & $2.25M$ & \numprint{2} & \numprint{33745} \\
    & \numprint{130710} & $1.94M$ & \numprint{3} & \numprint{11604} \\
    & \numprint{110109} & $1.83M$ & \numprint{4} & \numprint{9256} \\
    & \numprint{94664} & $1.77M$ & \numprint{5} & \numprint{4262} \\
    & \numprint{87113} & $1.70M$ & \numprint{6} & \numprint{5796} \\ 
    & \numprint{78142} & $1.62M$ & \numprint{7} & \numprint{3213} \\
    & \numprint{73070} & $1.57M$ & \numprint{8} & \numprint{2449} \\\hline
    coAuthorsDBLP & \numprint{299067} & \numprint{977676} & \numprint{1} & \numprint{45242} \\\hline
    cs4 & \numprint{22499} & \numprint{43858} & \numprint{2} & \numprint{2} \\\hline
    delaunay\_n17 & \numprint{131072} & \numprint{393176} & \numprint{3} & \numprint{1484} \\ \hline
    fe\_ocean & \numprint{143437} & \numprint{409593} & \numprint{1} & \numprint{40} \\\hline
    kron-logn16 & \numprint{55319} & $2.46M$ & \numprint{1} & \numprint{6325} \\\hline
    luxembourg & \numprint{114599} & \numprint{239332} & \numprint{1} & \numprint{23077} \\\hline
    vibrobox & \numprint{12328} & \numprint{165250} & \numprint{8} & \numprint{625} \\\hline
    wikipedia & \numprint{35579} & \numprint{495357} & \numprint{1} & \numprint{2172} \\
    \hline     
  \end{tabular}
\end{table}
\vfill
\pagebreak

\begin{table}[H]
  \caption{Set of large graphs from various sources\label{tab:large}}
  \begin{tabular}{|r|r|r|r|r|}
    \hline
    Name & $n$ & $m$ & $\lambda$ & $n^*$ \\ \hline\hline
    amazon-2008 & \numprint{735323} & $3.52M$ & \numprint{1} & \numprint{82520}\\
    & \numprint{649187} & $3.42M$ & \numprint{2} & \numprint{50611}\\
    & \numprint{551882} & $3.18M$ & \numprint{3} & \numprint{35752}\\
    & \numprint{373622} & $2.12M$ & \numprint{5} & \numprint{19813}\\
    & \numprint{145625} & \numprint{582314} & \numprint{10} & \numprint{64657}\\\hline
    coPapersCiteseer & \numprint{434102} & $16.0M$ & \numprint{1} & \numprint{6372} \\
    & \numprint{424213} & $16.0M$ & \numprint{2} & \numprint{7529} \\
    & \numprint{409647} & $15.9M$ & \numprint{3} & \numprint{7495} \\
    & \numprint{379723} & $15.5M$ & \numprint{5} & \numprint{6515} \\
    & \numprint{310496} & $13.9M$ & \numprint{10} & \numprint{4579} \\\hline
    eu-2005 & \numprint{862664} & $16.1M$ & \numprint{1} & \numprint{52232} \\
    & \numprint{806896} & $16.1M$ & \numprint{2} & \numprint{42151} \\
    & \numprint{738453} & $15.7M$ & \numprint{3} & \numprint{21265} \\
    & \numprint{671434} & $13.9M$ & \numprint{5} & \numprint{18722} \\
    & \numprint{552566} & $11.0M$ & \numprint{10} & \numprint{23798} \\\hline
    hollywood-2009 & $1.07M$ & $56.3M$ & \numprint{1} & \numprint{11923} \\
    & $1.06M$ & $56.2M$ & \numprint{2} & \numprint{17386} \\
    & $1.03M$ & $55.9M$ & \numprint{3} & \numprint{21890} \\
    & \numprint{942687} & $49.2M$ & \numprint{5} & \numprint{22199} \\
    & \numprint{700630} & $16.8M$ & \numprint{10} & \numprint{19265} \\\hline
    in-2004& $1.35M$ & $13.1M$ & \numprint{1} & \numprint{278092} \\
    & \numprint{909203} & $11.7M$ & \numprint{2} & \numprint{89895} \\
    & \numprint{720446} & $9.2M$ & \numprint{3} & \numprint{45289} \\
    & \numprint{564109} & $7.7M$ & \numprint{5} & \numprint{33428} \\
    & \numprint{289715} & $5.1M$ & \numprint{10} & \numprint{12947} \\\hline
    uk-2002 & $18.4M$ & $261.6M$ & \numprint{1} & $2.5M$ \\
    & $15.4M$ & $254.0M$ & \numprint{2} & $1.4M$ \\
    & $13.1M$ & $236.3M$ & \numprint{3} & \numprint{938319} \\
    & $10.6M$ & $207.6M$ & \numprint{5} & \numprint{431140} \\
    & $7.6M$ & $162.1M$ & \numprint{10} & \numprint{298716} \\
    & \numprint{657247} & $26.2M$ & \numprint{50} & \numprint{24139} \\
    & \numprint{124816} & $8.2M$ & \numprint{100} & \numprint{3863} \\
    \hline
  \end{tabular}
\end{table}
\end{appendix}

\end{document}